\documentclass[letterpaper, 10 pt, conference]{ieeeconf}

\usepackage[usenames,dvipsnames,svgnames,table]{xcolor}
\usepackage{amssymb,latexsym,amsfonts,amsmath}
\usepackage{graphicx}
\usepackage{eso-pic}
\usepackage{epsfig}
\usepackage{mathtools}
\usepackage{tikz}
\usetikzlibrary{automata}
\usetikzlibrary{shapes}
\usetikzlibrary{calc,shapes,arrows}
\usepackage{mdwlist}
\usepackage{refcount}
\usepackage{comment}

\usepackage{cite}

 \let\labelindent\IEEElabelindent
\usepackage{enumitem}
\usepackage{subcaption}
\usepackage{keyval,trig}
\usepackage{amssymb,dsfont}
\usepackage{mathtools,mathrsfs}

\definecolor{mediumblue}{rgb}{0.0, 0.0, 0.8}
\definecolor{mediumcandyapplered}{rgb}{0.0, 0.0, 0.8}
\definecolor{nazar}{rgb}{0.7, 0.5, 0.9}
\makeatletter
\let\NAT@parse\undefined
\makeatother
\usepackage[colorlinks=true, citecolor=mediumcandyapplered, linkcolor=mediumblue, urlcolor = nazar, final]{hyperref}

\usepackage{algorithm}
\usepackage{algorithmic}

\usepackage{tikz}
\usetikzlibrary{calc,shapes,arrows}

\usepackage{xspace}

\newcounter{theorem}
\newcounter{definition}
\newcounter{lemma}
\newcounter{claim}
\newcounter{problem}
\newcounter{proposition}
\newcounter{corollary}
\newcounter{construction}
\newcounter{example}
\newcounter{xca}
\newcounter{comments}
\newcounter{remark}
\newcounter{assumption}

\newtheorem{theorem}[theorem]{Theorem}

\newtheorem{definition}[definition]{Definition}

\newtheorem{remark}[remark]{Remark}
\newtheorem{assumption}[assumption]{Assumption}

\usepackage[many]{tcolorbox}
\usetikzlibrary{calc}
\tcbuselibrary{skins}

\newtcolorbox{resp}[1][]{%
	enhanced jigsaw,%
	colback=gray!5!white,%
	colframe=gray!80!black,%
	size=small,%
	boxrule=1pt,%
	halign title=flush center,%
	coltitle=black,%
	breakable,%
	drop shadow=black!50!white,%
	attach boxed title to top left={xshift=1cm,yshift=-\tcboxedtitleheight/2,yshifttext=-\tcboxedtitleheight/2},%
	minipage boxed title=3cm,%
	boxed title style={%
		colback=white,%
		size=fbox,%
		boxrule=1pt,%
		boxsep=2pt,%
		underlay={%
			\coordinate (dotA) at ($(interior.west) + (-0.5pt,0)$);
			\coordinate (dotB) at ($(interior.east) + (0.5pt,0)$);
			\begin{scope}[gray!80!black]
				\fill (dotA) circle (2pt);
				\fill (dotB) circle (2pt);
			\end{scope}
		}%
	},%
	#1%
}

\xdefinecolor{MATLABblue}{rgb}{0.0,0.45,.74}
\xdefinecolor{MATLABred}{rgb}{0.85,0.33,.1}

\newcommand{\R}{{\mathbb{R}}}

\IEEEoverridecommandlockouts
\usepackage{graphicx}
\usepackage{epstopdf}
\usepackage{adjustbox}
\newcommand{\greensquare}{\tikz\fill[green!70!white] (0,0) rectangle (2mm,2mm);}

\newcommand{\reddsquare}{\tikz\fill[red!30!white] (0,0) rectangle (2mm,2mm);}

\usepackage[normalem]{ulem}

\newcommand{\EE}{\mathds{E}}
\newcommand{\PP}{\mathds{P}}

\title{Safety Controller Synthesis for Stochastic Networked Systems under Communication Constraints}
\author{Omid Akbarzadeh$^1$, Mohammad H. Mamduhi$^2$, and Abolfazl Lavaei$^1$
\thanks{$^1$O. Akbarzadeh and A. Lavaei are with the School of Computing, Newcastle University, United Kingdom. Emails: \texttt{\{omid.akbarzadeh,abolfazl.lavaei\}@newcastle.ac.uk}.}
\thanks{$^2$M. H. Mamduhi is with the School of Computer Science, University of Birmingham, United Kingdom. Email:   \texttt{m.h.mamduhi@bham.ac.uk}.}}
\begin{document}
\maketitle

\begin{abstract}
	This paper develops a framework  for synthesizing safety controllers for discrete-time stochastic linear control systems (dt-SLS) operating under communication imperfections. The control unit is remote and communicates with the sensor and actuator through an imperfect wireless network. We consider a constant \emph{delay} in the sensor-to-controller channel (uplink), and \emph{data loss} in both sensor-to-controller and controller-to-actuator (downlink) channels. In our proposed scheme, data loss in each channel is modeled as an independent Bernoulli-distributed random process. To systematically handle the uplink delay, we first introduce an \emph{augmented} discrete-time stochastic linear system (dt-ASLS) by concatenating all states and control inputs that sufficiently represent the state-input evolution of the original dt-SLS under the delay and packet loss constraints. We then leverage control barrier certificates for dt-ASLS to synthesize a controller that ensures the stochastic safety of dt-SLS, guaranteeing that all trajectories remain outside unsafe regions with a quantified probabilistic bound. Our approach translates safety constraints into matrix inequalities, leading to an optimization problem that eventually quantifies the probability of satisfying the safety specification in the presence of communication imperfections. We validate our results on an RLC circuit subject to both constant delay and probabilistic data loss.
\end{abstract}

\section{Introduction}
\textbf{Motivation.}
Safety analysis in stochastic networked control systems has become a crucial design aspect due to the increasing application of such systems in safety-critical domains such as autonomous vehicles, networked robotics, and industrial automation. In many such systems, the controller is located away from the plant (cf. Fig.~\ref{fig: system}) for both practical and safety reasons. Representative examples include process industries, where controllers are housed in dedicated control rooms outside hazardous zones, and offshore or remote assets (\emph{e.g.}, wind turbines, pipelines) controlled from onshore facilities. In these settings, closed-loop operation depends on wireless communication channels, which are inherently subject to imperfections such as delays and data loss. Even a fixed delay in the feedback loop can degrade performance or destabilize a system, while data loss in sensor or actuator signals can further disrupt the closed-loop control performance~\cite{LQGdealypacketloss}. Hence, designing controllers that formally ensure a networked system remains within safe operational bounds under communication issues, such as delays and stochastic uncertainties, is challenging.

To address the challenges of safe controller synthesis for stochastic control systems, the existing literature has primarily focused on the promising approach of \emph{control barrier certificates (CBCs)}, initially introduced in~\cite{Pranja}. Similar to Lyapunov functions, CBCs are designed to satisfy specific conditions on both the function and its evolution along the system trajectories. By defining an initial level set of CBCs from a given set of initial states, these CBCs effectively separate unsafe regions from all admissible system evolutions, thereby providing (probabilistic) safety guarantees. This technique is widely applied for the formal verification and synthesis of both deterministic \cite{borrmann2015control,ames2019control} and stochastic dynamical systems \cite{lavaei2024scalable,nejati2024context,zaker2024compositional}.

\begin{figure}
	\centering
	\includegraphics[width=0.8\linewidth]{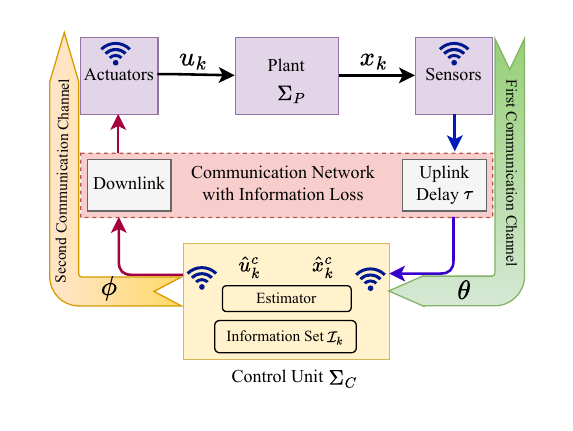}
	\caption{Schematic of the networked control system architecture, where the controller interacts with the plant through communication channels subject to uplink delay and data loss. Data loss in the uplink and downlink channels is modeled by independent Bernoulli variables $\theta$ and $\phi$, respectively.}
	\label{fig:system_architecture}
	\label{fig: system}
	\vspace{-3mm}
\end{figure}

While CBCs have been effectively used to enforce safety constraints in stochastic control systems, they typically assume a continuously closed feedback loop. However, in the \emph{simultaneous} presence of delays and data loss, the loop may be intermittently broken: control signals might fail to reach the actuators on time, and control actions may be stalled due to missing sensor measurements. Although several studies have addressed the safety analysis of control systems with time delay~\cite{Amesdelay2023,ahmadi2019safe,REN2022}, they typically assume a \emph{reliable} feedback loop, while ignoring data loss in communication channels. As a result, these methods are inadequate for scenarios where \emph{both} delay and data loss jointly impact control synthesis and closed-loop behavior.

The challenges of control systems operating over wireless communication channels, where data loss and delays are prevalent, have been investigated in the literature (see \emph{e.g.},~\cite{4118454,4118476, ip-cta_20050178}).  Approaches like the Glossy architecture\cite{5779066} improve reliability via synchronous flooding, while studies on Kalman filtering with intermittent observations~\cite{6004816} show that even minor communication failures can degrade control accuracy. Optimal control strategies have been developed to address these issues, including linear quadratic Gaussian control under delay-dependent constraints~\cite{8405590} and models incorporating traffic-correlated delays and losses~\cite{9462479}. Broader surveys are available in~\cite{4118465}, and the work~\cite{maggio_et_al} has also examined stability under repeated deadline misses, emphasizing the need for timeliness in ensuring robustness. The stability analysis of control systems over wireless networks has also been addressed in~\cite{Trimpe-stability1,Trimpe-stability2}, where the augmentation approach used in our framework is inspired by~\cite{Trimpe-stability1}. Nevertheless, none of the above studies have considered the \emph{coupled effects} of delay and data loss in the context of \emph{safety analysis} for control systems over wireless networks, as we addressed within this work.

\textbf{Key Contribution.}
Inspired  by the above challenges, we develop a formal framework for synthesizing safe controllers for dt-SLS under communication constraints. Specifically, we consider \emph{delays} in the sensor-to-controller (uplink) channel and probabilistic \emph{data loss} in both the uplink and controller-to-actuator (downlink) channels (see Fig.~\ref{fig: system}). Our approach handles communication imperfections by introducing an augmented system representation (dt-ASLS) that captures the effects of uplink delays and stochastic data loss. We leverage CBCs within this augmented framework to provide probabilistic safety guarantees by formulating safety constraints as matrix inequalities. Unlike our prior work~\cite{Akbarzadeh}, which assumes no communication delay, the proposed framework accounts for the coupled effects of uplink delays and data loss, making the problem significantly more challenging.

\section{Control System with Communication Model}\label{sec:model}
\subsection{Notation and Preliminaries}
We denote the sets of real, non-negative and positive real numbers, respectively, by $\mathbb{R}$, $\mathbb{R}_{ 0}^+$, and $\mathbb{R}^{+}$. Sets of non-negative and positive integers are denoted by $\mathbb{N}$ and~$\mathbb{N}^{+}$, respectively. Given $N$ vectors $x_i \in \mathbb{R}^{n_i}$, $x\!=\![x_1;\dots;x_N]$ is its concatenated \emph{column} vector. Trace of a matrix $A \in \mathbb{R}^{N \times N}$ with diagonal elements $a_1, \dots , a_N $ is written as $\text{Tr}(A)\!=\! \sum_{i=1}^N a_i$. An identity matrix of dimension~$n$, and a zero matrix of dimension $n\times m$ are denoted, respectively, by  $\mathbf{I}_n$ and $\mathbf{0}_{n\times m}$. A normal distribution with the mean value $\mu$ and the variance $\sigma$ is denoted by $\mathcal{N}(\mu,\sigma)$.
For a \emph{symmetric} matrix $\mathcal{P}$, $\mathcal{P} \succ 0$ ($\mathcal{P} \succeq 0$) denotes $\mathcal{P}$ is positive definite (positive semi-definite). The transpose of a matrix $A$ is written as $A^\top$. The Euclidean norm of a vector $x \in \mathbb{R}^n$ is expressed as $\Vert x \Vert$. The empty set is denoted by $\emptyset$. A diagonal matrix in $\R^{n\times{n}}$ with  diagonal entries $(a_1,\ldots,a_n)$ is signified by $\mathsf{diag}(a_1,\ldots,a_n)$. The variance of a random variable \(a\) is denoted by \(\textsf{Var}[a]\). Given a system $\Sigma$ and a property $\Upsilon$, $\Sigma \models \Upsilon$ denotes $\Sigma$ satisfies $\Upsilon$.

By $(\Omega, \mathcal{F}_{\Omega}, \PP_{\Omega})$, we denote a probability space, where $\Omega$ is the sample space, $\mathcal{F}_{\Omega}$ is a sigma-algebra on $\Omega$, and $\PP_{\Omega}$ is the probability measure. We assume random variables $\mathcal X$ are measurable functions, \emph{i.e.,} $\mathcal X:(\Omega,\mathcal F_{\Omega})\rightarrow (S_{\mathcal X},\mathcal F_{\mathcal X})$, such that $\mathcal X$ induces a probability measure on $(S_{\mathcal X},\mathcal F_{\mathcal X})$ as $Prob\{\mathcal A\} = \PP_{\Omega}\{\mathcal X^{-1}(\mathcal A)\}$, $\forall \mathcal A\in \mathcal F_{\mathcal X}$.
A topological space~$\mathbf{S}$ is said to be Borel, denoted by $\mathcal B(\mathbf{S})$, if it is homeomorphic to a Borel subset of a Polish space, \textit{i.e.}, a separable and metrizable space.

\subsection{Networked Control System}
We consider a discrete-time stochastic linear control system (dt-SLS) with the following dynamics
\begin{equation}\label{Plant}
	x_{k+1} = A x_{k} + B u_{k} + w_{k},\;\; k\in \mathbb N,
\end{equation}
where \( x_k \in X \) and \( u_k \in U \), with \( X \subseteq \mathbb{R}^n \) and  \( U \subseteq \mathbb{R}^m\)  being Borel spaces representing the state and input sets, respectively. 
The matrices \( A \in \mathbb{R}^{n \times n} \) and \( B \in \mathbb{R}^{n \times m} \) 
are the state and input matrices of the system. 
Additionally, \( w_k \) denotes a sequence of independent and identically distributed (i.i.d.)
random variables from a normal
distribution~$\mathcal{N}(0,\sigma_w)$, representing the system's \emph{process} noise.

The collocated control and estimation unit of the system is assumed to be remote, communicating with the sensing and actuator stations via wireless communication channels. We assume the communication channel introduces both \emph{delay} and \emph{data loss} during data transmission. To account for communication issues, 
we introduce a constant delay \(\tau \in \mathbb{N}\) in the uplink channel, \emph{i.e.}, between the sensor and the controller. Additionally, to model data loss in the uplink and downlink channels, we consider two \emph{mutually independent} random processes \(\theta_k\) and \(\phi_k\), $k\in \mathbb{N}$, where the corresponding binary random variables are selected from Bernoulli distributions.
Specifically, $\theta_k = 1$ indicates that the state information sent to the controller at time $k$ is successfully received, while $\theta_k = 0$ denotes a data loss at that time. Note that due to the uplink delay, if a data packet is not dropped, it is received by the controller after $\tau$ time steps. Similarly, $\phi_k = 1$ indicates that the actuator successfully receives the control input at time $k$, while $\phi_k = 0$ means the control input is not received at that time. Furthermore, \(p_{\theta}\) and \(q_{\phi}\) denote, respectively, the probabilities of \emph{successful} data transmissions in uplink and downlink channels, \emph{i.e.}, \(\PP[\theta_k=1]=p_{\theta}\), and \(\PP[\phi_k=1]=q_{\phi}\), for any $k \in \mathbb N$. 

We now define the networked system model $\Sigma$, based on the dt-SLS dynamics in~\eqref{Plant}, incorporating communication delay and data loss.

\begin{definition}[\textbf{dt-SLS-$\Sigma$}]\label{dt-SLS}
	The networked system, comprising the dt-SLS plant in~\eqref{Plant} and a remote control unit communicating over wireless channels with delay and data loss, can be formally defined by the following tuple
	\begin{equation}\label{dt-SLSCC-tuple}
		\Sigma=(X,U,A,B,w, \tau, \theta, \phi).
	\end{equation}    
	We represent $\Sigma$ using two state-space models: one for the plant, denoted by $\Sigma_P$, and one for the control unit, represented by $\Sigma_C$. The dt-SLS-$\Sigma$ then evolves according to the following dynamics:
	\begin{align}\label{dt-SLSCC}
		&\Sigma_P:
		\begin{cases} 
			x_{k+1}=Ax_k+Bu_k+w_k,\\
			u_k=\phi_k \hat{u}_k^c + (1-\phi_k)u_{k-1},
		\end{cases}\\
		&\Sigma_C: 
		\begin{cases}\label{dt-SLSCC1}
			\hat{x}^c_{k+1}=A\hat{x}^c_k+B\hat{u}_k^c,\\
			\hat{u}_k^c = \mathcal{F}\hat{x}^c_k,
		\end{cases}
	\end{align}
	where $\mathcal{F} \in \mathbb{R}^{n \times n}$ is a control matrix, \(\hat{x}_{k}^c \in X\) (the estimate of plant state in~\eqref{control-s-f}) and \(\hat{u}_{k}^c \in U\) are the state and input of the controller at time $k$, and $u_k$ is the control signal applied to the \emph{actuator} at time $k$. In fact, $\hat{u}_k^c$ is the control signal computed at the control unit at time $k$, and according to \eqref{dt-SLSCC}, if $\phi_k=1$, then $u_k=\hat{u}_k^c$. The initial state, denoted by $x_0 \sim \mathcal{G}$ with $x_0 \in X$, is drawn from a known distribution $\mathcal{G}$ with finite variance.
	We denote by $\mathsf{x}_k$ the solution process of dt-SLS-$\Sigma$ at time $k \in \mathbb{N}$, under input and process noise trajectories $u(\cdot)$ and $w(\cdot)$, starting from an initial state $x_0\in X$.
\end{definition}

\begin{remark} The dynamics in \eqref{dt-SLSCC} and \eqref{dt-SLSCC1} capture the interaction between the plant and remote controller over unreliable wireless channels. We represent \(\Sigma\) into \(\Sigma_P\) and \(\Sigma_C\) to separate the plant dynamics from the controller logic, which is affected by delays and data loss. Due to the uplink delay, the control unit always receives delayed state information. Specifically, at time $k$, the most recent state available to the controller is $x_{k-\tau}$, provided the uplink transmission at time $k-\tau$ was successful (\emph{i.e.}, $\theta_{k-\tau} = 1$). The control input $u_k$ applied to the plant is $\hat{u}_k^c$ if $\phi_k = 1$, or the previous input $u_{k-1}$ if $\phi_k = 0$, \emph{i.e.,} a zero-order hold mechanism is used in the event of downlink data loss (cf.~\eqref{dt-SLSCC}). The controller evolves based on an internal state \(\hat{x}_k^c\), representing the estimate of plant state at time $k$ based on the available information and applies a state-feedback control input \(\mathcal{F}\hat{x}_k^c\) (cf.~\eqref{dt-SLSCC1} and~\eqref{control-s-f}).
\end{remark}

The state \( x_k \) of the plant \( \Sigma_P \)~in~\eqref{dt-SLSCC} can be rewritten with \(\tau\)-step backward recursion~\cite{9462479}, as follows:
\begin{align}\label{recursion}
	x_k &= A^\tau x_{k-\tau} + \sum_{t=0}^{\tau-1} A^t B u_{k-t-1} + \sum_{t=0}^{\tau-1} A^t w_{k-t-1}.
\end{align}
Substituting the control input from~\eqref{dt-SLSCC}, given by $u_k = \phi_k \hat{u}_k^c + (1 - \phi_k) u_{k-1}$, into~\eqref{recursion} yields
\begin{align}\notag
	x_k&=A^{\tau} x_{k-\tau}+\sum_{t=0}^{\tau-1}\big(\phi_{k-t-1}A^tB\hat{u}_{k-t-1}^c\\\label{P-State-delay} &~~~+(1-\phi_{k-t-1})A^tB u_{k-t-2} +A^tw_{k-t-1}\big).
\end{align}
Note that the controller does not have access to the variables \(\phi_{k}\)~in~\eqref{P-State-delay} since these occur on the actuator side and no feedback on \(\phi_k\) is provided to the controller. Consequently, when constructing the networked control system dt-SLS-$\Sigma$~in~Definition~\ref{dt-SLS}, the controller always assumes that the control input applied to the plant \(\Sigma_P\) is \(\hat{u}_k^c\) (\emph{i.e.,}~$\phi_k=1$).
Then, we define the controller's state at time \( k \) as $\hat{x}^c_k = \EE\big[{x}_{k}\,\big|\,\phi_k=1,\forall\, k\,,\mathcal{I}_k\big]$, where $\mathcal{I}_k$ is the set of available information at the control side at time $k$, defined as
\begin{align*}
	\mathcal{I}_k
	\;=\;
	\Bigl\{&A, B, x_0,\,\{\theta_0,\dots,\theta_{k-\tau}\},\\&\{\theta_0x_0,\dots,\theta_{k-\tau}x_{k-\tau}\},\,
	\{\hat{u}^c_{0},\cdots,\hat{u}^c_{k-1}\}
	\Bigr\}.
\end{align*}
By leveraging \eqref{P-State-delay}, we then have
\begin{align}\nonumber
	\hat {x}^c_k &= \EE\big[{x}_{k}\,\big|\,\phi_k=1,\forall\, k\,,\mathcal{I}_k\big]\\\label{State-Controller-Delayed}
	&=\EE\Big[\theta_{k-\tau}\Big(A^{\tau}x_{k-\tau} +\sum_{t=0}^{\tau-1}A^tB\hat{u}_{k-t-1}^c\\\nonumber&+\sum_{t=0}^{\tau-1}A^tw_{k-t-1}\Big)\!+\underbrace{(1-\theta_{k-\tau})(A\hat{x}^c_{k-1}\!+\!B\hat{u}_{k-1}^c)}_{(\star)}\,\big|\,\mathcal{I}_k\Big]\!.
\end{align}
We have formulated $(\star)$ based on the information set $\mathcal{I}_k$ available to the control unit at time $k$. If $\theta_{k-\tau} = 0$, the control unit does not receive the plant state $x_{k-\tau}$, which is the most recent state the controller could access at time $k$ due to the $\tau$-step delay. In this case, the control unit uses the model in~\eqref{dt-SLSCC1} to compute a model-based estimate of the state $\hat{x}_k^c$, which serves as the state of the controller. Since \(w_k\) is a sequence of i.i.d. random variables drawn from \(\mathcal{N}(0,\sigma_w)\), and \(\EE[\theta_{k-\tau}] = p_\theta\), we can rewrite \eqref{State-Controller-Delayed} as
\begin{align}\notag
	\hat {x}^c_k &=\EE\big[{x}_{k}\,\big|\,\phi_k=1,\forall\, k\,,\mathcal{I}_k\big]\\\notag&=p_\theta\Big(A^{\tau}x_{k-\tau}+\sum_{t=0}^{\tau-1}A^tB\hat{u}_{k-t-1}^c\Big)\\\label{control-s-f}&~~~+\!(1-p_\theta)(A
	\hat{x}^c_{k-1}+B\hat{u}_{k-1}^c).
\end{align}
Next, we provide the formal definition of safety for the described dt-SLS-$\Sigma$~in~Definition~\ref{dt-SLS}.

\begin{definition}[\textbf{Safety under delay and data loss}]\label{Safety}
	Given a dt-SLS-$\Sigma=(X,U,A,B,w, \tau, \theta, \phi)$~ as in~Definition~\ref{dt-SLS} subjected to delay and data loss, consider a safety specification $\Upsilon = (X_0, X_1, \mathcal T)$, where $X_0,X_1\subseteq X$ are its initial and unsafe sets, respectively. We assume $X_0\cap X_1 = \emptyset$, otherwise system is unsafe with a probability of 1. The dt-SLS-$\Sigma$ is said to be safe within time horizon $\mathcal{T} \in \mathbb{N}$, denoted by $\Sigma\models\Upsilon$, if all trajectories of $\Sigma$ starting from $X_{0}$ do not reach $X_1$ during  $\mathcal{T}$\!. Since trajectories
	of dt-SLS-$\Sigma$ are probabilistic, we are interested in
	computing $\PP \{\Sigma\models\Upsilon\}\ge 1-\xi_{\tau\theta\phi}$,~(cf.~Remark~\ref{Delay-contribution}) where $\xi_{\tau\theta\phi} \in(0,1]$.
\end{definition}
\begin{remark}\label{Delay-contribution}
	Although the effects of uplink delay~$\tau$ and probabilistic data loss in both uplink and downlink channels are not explicit in the safety specification~$\Upsilon$, they influence the lower bound of the safety probability $1-\xi_{\tau\theta\phi}$ (cf.~\eqref{matrix-inequlity}).
\end{remark}

\subsection{Augmented dt-SLS-$\Sigma$}
To formally analyze safety of the networked control system dt-SLS-$\Sigma$, we first introduce an augmented networked system representation that captures both the inputs and states of the dt-SLS-$\Sigma$ and serves as the foundation for defining the CBC, which follows next. The dynamics of $x_k$ and $\hat{x}_k^c$ are given in~\eqref{dt-SLSCC} and~\eqref{dt-SLSCC1}, while dynamics of $u_k$ and $\hat{u}_k^c$ can be defined as
\begin{align*}
	u_{k+1}&= \phi_{k+1} \mathcal{F}(A\hat{x}^c_k+B\hat{u}_k^c) + (1-\phi_{k+1})u_{k},\\
	\hat{u}_{k+1}^c &= \mathcal{F}\hat{x}^c_{k+1}=\mathcal{F}(A\hat{x}^c_k+B\hat{u}_k^c).
\end{align*}
Let us define \(\mathbf{{x}}_k = [{x}_k; {x}_{k-1}; \dots; {x}_{k-\tau}]\), and \(\mathbf{\hat{x}}_k^c = [\hat{x}_k^c; \hat{x}_{k-1}^c]\), where \(\mathbf{x}_k \in  \mathbb{R}^{\psi}\), and $\mathbf{\hat x}_k^c \in \mathbb{R}^{2n}$ with $\psi= n(\tau+1)$. Moreover, \({\mathbf{u}}_k = [u_{k};u_{k-1}; \dots; u_{k-{\tau}};u_{k-{\tau}-1}]\), and \(\hat{\mathbf{u}}_k^c = [\hat{u}_{k}^c;\hat{u}_{k-1}^c; \dots; \hat{u}_{k-{\tau}}^c]\) where \({\mathbf{u}}_k\in \mathbb{R}^{\varpi+m},\hat{\mathbf{u}}_k^c \in \mathbb{R}^{\varpi}\) and $\varpi= m(\tau+1)$. Additionally, we define a noise vector \(\mathbf{w}_k = [w_k;w_{k-1}; \dots; w_{k-{\tau}}]\), where \(\mathbf{w}_k \in \mathbb{R}^{\psi}\).
We then augment $\textbf{x}_k$, $\hat{\textbf{x}}_k^c$, $\mathbf{u}_k$, and $\hat{\textbf{u}}_k^c$ of our dt-SLS-$\Sigma$ as \(\mathcal{Z}_k = [\mathbf{x}_k;\mathbf{\hat x}_k^c;\mathbf{u}_k;\hat{\mathbf{u}}_k^c]\), where \(\mathcal{Z}_k \in \mathbb{R}^\kappa\) with \(\kappa = 2\varpi+\psi+2n+m\). The \emph{augmented dt-SLS}, denoted as dt-ASLS-$\Lambda$, evolves as
\begin{equation}\label{eq:augmented}
	\Lambda: \mathcal{Z}_{k+1} = \mathcal{A} \mathcal{Z}_k + \mathcal{D} \mathbf{w}_k, 
\end{equation}
where \(\mathcal{A} \in \mathbb{R}^{ \kappa \times \kappa}\) and \(\mathcal{D}\in \mathbb{R}^{\kappa \times \psi}\), with its equivalent representation
\begin{align}\label{aug-system}
	\!\!\Lambda\!&:\!\! 
	\begin{bmatrix} { \mathbf{x}}_{k+1}\\  \hat{ \mathbf{x}}_{k+1}^c\\ { \mathbf{u}}_{k+1}\\ \hat{ \mathbf{u}}_{k+1}^c \end{bmatrix} \!\!=\!\!
	\underbrace{\begin{bmatrix} \mathcal{A}_{11} & \mathcal{A}_{12} & \mathcal{A}_{13} & \mathcal{A}_{14} \\ \mathcal{A}_{21} & \mathcal{A}_{22} &\mathcal{A}_{23}&\mathcal{A}_{24}\\
			\mathcal{A}_{31} & \mathcal{A}_{32} & \mathcal{A}_{33}& \mathcal{A}_{34}\\
			\mathcal{A}_{41} & \mathcal{A}_{42} & \mathcal{A}_{43} & \mathcal{A}_{44}
	\end{bmatrix}}_{\mathcal{A}}\!
	\begin{bmatrix} { \mathbf{x}}_{k}\\  \hat{ \mathbf{x}}_{k}^c\\ { \mathbf{u}}_{k}\\ \hat{ \mathbf{u}}_{k}^c \end{bmatrix} \!\!+\!
	\mathcal{D} \mathbf{w}_k,\\\notag
	&\text{with}~
	\mathcal{D}  =\! \begin{bmatrix}  \mathbf{I}_{n}&{A} & \dots & {A}^\tau\\
		\boldsymbol{0}_{(\kappa -n) \times n} &\boldsymbol{0}_{(\kappa -n) \times n} & \dots & \boldsymbol{0}_{(\kappa -n) \times n}\end{bmatrix}_{\kappa \times \psi}\!\!\!\!\!\!\!\!\!\!\!\!,
\end{align}
where $\mathcal A_{ij},~ i,j\in \{1,\dots,4\}$ are defined as in~\eqref{A-elements} in the Appendix. Since $\phi_{k - t - 1}$ is an i.i.d.\ Bernoulli random variable for all $t \in \{0, 1, \dots, \tau - 1\}$, we simplify the notation by denoting $\phi_{k - t - 1}$ as $\phi$ and $1 - \phi_{k - t - 1}$ as $\bar{\phi}$. We used the \emph{padding} technique~\cite{maggio_et_al} in constructing matrix \(\mathcal{A} \) to explicitly represent the state variable transitions from one time instant to the next. The padding incorporates trivial equations \( x_i = x_i,\,\hat{u}^c_i = \hat{u}^c_i \), with \(\forall i \in [k - \tau + 1, k] \), \( \hat{x}_k^c = \hat {x}_k^c \), and \({u}_i = {u}_i\), with \(\forall i \in [k - \tau, k] \) ensuring continuity in the state representation of the augmented system dt-ASLS-$\Lambda$. 

We now decompose matrix $\mathcal{A}$ in~\eqref{aug-system} into a sum of multiple components. Specifically, $\mathcal{A}$ is expressed as a combination of a baseline matrix $\mathcal{A}_1$ and additional terms that depend on the Bernoulli-distributed variable $\phi$. This decomposition allows us to analyze the dynamics of the augmented system~dt-SLS-$\Lambda$ by isolating the effect of $\phi$. Furthermore, we use the transformation $\phi = q_\phi (1 - \zeta_\phi)$, where the new variable \(\zeta_\phi\) has a two-point distribution and takes one of two possible values: either \(1\) or \(1 - \frac{1}{q_\phi}\), with $\EE[\zeta_\phi]=0$ and $\mathsf{Var}[\zeta_\phi]=\EE[\zeta_\phi^2]=\frac{1-q_\phi}{q_\phi}$.
This transformation simplifies the computation by allowing the decomposition of $\mathcal{A}$ as 
\begin{equation}\label{decopmposition}
	\mathcal{A} = \mathcal{A}_{1} +  \mathcal{A}_{2} + \zeta_\phi\mathcal{A}_{3},
\end{equation}
where
	\begin{align}\notag
		\mathcal{A}_1 &\!\!=\!\! \begin{bmatrix} \mathcal{A}_{11} \!\!&\!\! \mathcal{A}_{12} \!\!&\!\! \mathcal{A}_{13}^a \!\!&\!\! \mathbf{0}_{\psi \times \varpi} \\  {\mathcal{A}}_{21} \!\!&\!\! {\mathcal{A}}_{22} \!\!&\!\!\mathcal{A}_{23} \!\!&\!\!\mathcal{A}_{24}\\
			\!\mathbf{0}_{(\varpi+m) \times \psi} \!\!&\!\! \mathbf{0}_{(\varpi +m) \times 2n} \!\!&\!\!\mathcal{A}^a_{33} \!\!&\!\!\mathbf{0}_{(\varpi+m) \times \varpi}\!\!\!\\
			\mathbf{0}_{\varpi \times \psi} \!\!&\!\! \mathbf{0}_{\varpi \times 2n} \!\!&\!\!\mathbf{0}_{\varpi \times (\varpi+m)} \!\!&\!\!\mathbf{0}_{\varpi \times \varpi}
		\end{bmatrix}\!\!,\\\notag
	\end{align}
	\begin{align}\notag
		\mathcal{A}_2 &\!\!=\!\! \begin{bmatrix} \mathbf{0}_{\psi \times \psi} \!&\! \mathbf{0}_{\psi \times 2n} \!&\! (1-q_\phi) \mathcal{A}_{13}^b \!&\! q_\phi\mathcal{A}_{14}^a \\  \mathbf{0}_{2n \times \psi} \!&\! \mathbf{0}_{2n \times 2n} \!&\! \mathbf{0}_{2n \times (\varpi+m)} \!&\! \mathbf{0}_{2n \times \varpi}\\
			q_\phi\mathcal{A}_{31}^a\!&\! q_\phi\mathcal{A}_{32}^a \!&\! (1-q_\phi)\mathcal{A}_{33}^b \!&\! q_\phi\mathcal{A}_{34}^a \\
			{\mathcal{A}}_{41} \!&\! {\mathcal{A}}_{42} \!&\!\mathcal{A}_{43} \!&\!\mathcal{A}_{44}
		\end{bmatrix}\!\!,\\\label{Newjh}
		\mathcal{A}_3 &\!\!=\!\! \begin{bmatrix} \mathbf{0}_{\psi \times \psi} \!&\! \mathbf{0}_{\psi \times 2n} \!&\! q_\phi \mathcal{A}_{13}^b \!&\! -q_\phi\mathcal{A}_{14}^a \\  \mathbf{0}_{2n \times \psi} \!&\! \mathbf{0}_{2n \times 2n} \!&\! \mathbf{0}_{2n \times (\varpi+m)} \!&\! \mathbf{0}_{2n \times \varpi}\\
			-q_\phi\mathcal{A}_{31}^a\!&\! -q_\phi\mathcal{A}_{32}^a \!&\! q_\phi\mathcal{A}_{33}^b \!&\! -q_\phi\mathcal{A}_{34}^a\\
			\mathbf{0}_{\varpi \times \psi} \!&\! \mathbf{0}_{\varpi \times 2n} \!&\!\mathbf{0}_{\varpi \times (\varpi+m)} \!&\!\mathbf{0}_{\varpi \times \varpi}
		\end{bmatrix}\!\!,
	\end{align}
and the new entries $\mathcal A^a_{ij}$ and $\mathcal A^b_{ij}$ can be found in~\eqref{New-entries} in the Appendix. In particular, the new transformation $\phi = q_\phi (1 - \zeta_\phi)$ relaxes the condition in~\eqref{matrix-inequlity} by eliminating cross terms arising from the multiplication of $\mathcal{A}_3$ with $\mathcal{A}_2$ and $\mathcal{A}_1$, given that $\EE[\zeta_\phi] = 0$ (cf. proof of Theorem~\ref{main-theorem}). It also facilitates the treatment of bilinearity, as will be discussed in Subsection~\ref{Bilinearity}, enabling us to derive the linear conditions presented in~\eqref{Linear-conditins}.
With the augmented networked control system formulated, we now introduce the CBC for the dt-ASLS-$\Lambda$ in~\eqref{aug-system} in the next section.

\section{Control Barrier Certificate}\label{sec:CBC}
\begin{definition}[CBC]\label{barrier}
	Consider a dt-SLS-$\Sigma=(X,U,A,B,w, \tau, \theta, \phi)$ in~Definition~\ref{dt-SLS}, with the safety specification $\Upsilon = (X_0, X_1, \mathcal T)$, and \(\hat{x}_{k}^c\), \(\hat{u}_{k}^c\) as the state and input of the controller in~\eqref{dt-SLSCC}. A function $\mathbb{B}\!: X \!\times\! X \!\times\! U \!\times\! U \to \mathbb{R}^{+}_{0}$ is said to be a control barrier certificate (CBC) for the augmented dt-ASLS-$\Lambda$ if there exist $c,\beta,\eta \in \mathbb{R}^{+}$ with $ \beta > \eta  $, such that
	\begin{subequations}
		\begin{align}\label{con}
			&\forall \mathcal{Z} \!\in\! X_{0} \!\times\! X_0\!\times\! U\!\times\! U\!:\quad \mathbb{B}(\mathcal{Z}) \leq \eta,\\\label{con1}
			&\forall \mathcal{Z} \!\in\! X_{1}\!\times\! X_1\!\times\! U\!\times\! U\!:\quad \mathbb{B}(\mathcal{Z}) \geq \beta,\\\label{con3}
			&\forall \mathcal{Z} \!\in\! X\!\times\! X \!\times\! U\!\times\! U\!:\quad \EE\Big[ \mathbb{B}(\mathcal{Z}_{k+1}) \big| \mathcal{Z}_k \Big] \!\leq \mathbb{B}(\mathcal{Z}_k)+c.
		\end{align}
	\end{subequations}	
\end{definition}
\begin{remark}[\textbf{On the effect of $\tau$,~$\theta$,~$\phi$}]
	The CBC condition~\eqref{con3} is influenced by the uplink delay~$\tau$ and probabilistic data loss in both uplink ($\theta$) and downlink ($\phi$) channels, as these factors explicitly appear in the dt-ASLS-$\Lambda$ described in~\eqref{eq:augmented}. Consequently, these communication issues directly affect the CBC $\mathbb{B}$ and inevitably influence the values of the level sets $\eta$ \eqref{con}, $\beta$ \eqref{con1}, and the constant $c$. These parameters collectively quantify the lower bound on the safety probability~$1-\xi_{\tau\theta\phi}$, as outlined in Theorem~\ref{Kushner}.
\end{remark}

\begin{remark}[\textbf{On downlink delays}]\label{downlink-delay}
Our framework can be readily extended to handle downlink delays using a similar approach as for the uplink. However, for clarity of presentations, we leave this extension to future works.
\end{remark}

\subsection{Computation of CBC and Safety Controller}
We consider the structure of our CBC to be quadratic as $\mathbb{B}(\mathcal{Z}) = \mathcal{Z}^{\top}\mathcal{P}\mathcal{Z}$, where $\mathcal{P}\in\mathbb{R}^{\kappa \times\kappa} $, with \(\kappa = 2\varpi+\psi+2n+m\), and $\mathcal{P} \succ 0$. We then propose a framework for designing a CBC and its corresponding safety controller to enforce safety requirements for the dt-SLS-$\Sigma$. 
To achieve this, we first raise the following assumption.
\begin{assumption} \label{assum}
	Assume for some given $p_\theta$, $q_\phi \in (0,1)$, there exist matrices $\mathcal{P} \succ 0$ and $\mathcal{F}$ that fulfill the matrix inequality
	\begin{align}\notag
		\mathcal{A}_{1}^\top \mathcal{P}\mathcal{A}_{1} +\mathcal{A}_{1}^\top \mathcal{P}\mathcal{A}_{2}&+  \mathcal{A}_{2}^\top \mathcal{P}\mathcal{A}_{2} + \mathcal{A}_{2}^\top \mathcal{P}\mathcal{A}_{1}\\\label{matrix-inequlity}&+(\frac{1-q_\phi}{q_\phi})\mathcal{A}_{3}^\top  \mathcal{P}\mathcal{A}_{3} - \mathcal{P} \preceq 0.
	\end{align}
\end{assumption}
\vspace{0.2cm}
\begin{remark}
	The maximum delay and data loss that a control system can tolerate depend majorly on the eigenvalues of the state matrix $A$, which is reflected in the matrix inequality~\eqref{matrix-inequlity}. Generally, state matrices with larger eigenvalues severely limit the tolerable delay and data loss, whereas state matrices with eigenvalues placed well inside the unit circle allow for larger delays and higher data loss while fulfilling Assumption~\ref{assum}, \cite{HU20071243,Gaudette}.
\end{remark}

\subsection{Handling Bilinearity}\label{Bilinearity} 
	Before presenting the main results of this work under Assumption~\ref{assum}, we note that there is a bilinearity between $\mathcal{P} $ and $\mathcal{F} $ in the matrix inequality~\eqref{matrix-inequlity}. 
	To resolve it, we define weight coefficients $\delta_i \in(0,1), i \in \{1,2,3,4,5\},$ such that $\sum_{i=1}^5\delta_i=1$, and
	\begin{subequations}\label{matrix-inequality1}
		\begin{align}
			\mathcal{A}_1^\top \mathcal{P} \mathcal{A}_1 - \delta_1\mathcal{P} \preceq 0, & \quad \label{matrix-inequality1a} \\ 
			\mathcal{A}_1^\top \mathcal{P} \mathcal{A}_2 - \delta_2\mathcal{P} \preceq 0, & \quad \label{matrix-inequality1b} \\ 
			\mathcal{ A}^\top_{2} \mathcal{P}  \mathcal{ A}_{2} -\delta_3 \mathcal{P} \preceq 0, & \quad \label{matrix-inequality1c} \\ 
			\mathcal{ A}^\top_{2} \mathcal{P}  \mathcal{A}_{1} -\delta_4 \mathcal{P} \preceq 0, & \quad \label{matrix-inequality1d}\\
			(\frac{1-q_\phi}{q_\phi})\mathcal{ A}^\top_{3} \mathcal{P}  \mathcal{A}_{3} -\delta_5 \mathcal{P} \preceq 0. & \quad \label{matrix-inequality1E}
		\end{align}
	\end{subequations}
	The matrix inequalities in~\eqref{matrix-inequality1} are equivalent to the matrix inequality in~\eqref{matrix-inequlity} (see proof of Theorem~\ref{main-theorem}). Consequently, satisfying all the individual matrix inequalities in~\eqref{matrix-inequality1} ensures the satisfaction of~\eqref{matrix-inequlity}, thereby validating the overall condition. To proceed, we first solve the matrix inequality in~\eqref{matrix-inequality1a} for $\mathcal{P}$ and $\delta_1$. To avoid the mild bilinearity between $\mathcal{P}$ and the scalar value $\delta_1$, we initialize $\delta=\{\delta_1,\delta_2,\delta_3,\delta_4,\delta_5\}$ (cf. Algorithm~\ref{Alg1}). Once $\mathcal{P}$ is obtained by ensuring the feasibility of~\eqref{matrix-inequality1a}, we focus on solving the remaining matrix inequalities in~\eqref{matrix-inequality1b}-\eqref{matrix-inequality1E} to compute $\mathcal{F}$, $\delta_2$, $\delta_3$, $\delta_4$, and $\delta_5$. Given that
	\begin{align*}
		\mathcal{A}_1^\top \mathcal{P} \mathcal{A}_2 = \mathcal{A}_1^\top \mathcal{P}  \mathcal{P}^{-1} \mathcal{P} \mathcal{A}_2 = (\mathcal{P} \mathcal{A}_1)^\top \mathcal{P}^{-1} (\mathcal{P} \mathcal{A}_2),
	\end{align*}
	by applying Schur complement~\cite{zhang2006schur} on the matrix inequality in~\eqref{matrix-inequality1b}, one can confirm that
	\begin{subequations}\label{Linear-conditins}
		\begin{align}\notag
			&\mathcal{P} \succ 0, \mathcal{A}_1^\top  \mathcal{P} \mathcal{A}_2  - \delta_2 \mathcal{P} \preceq 0 \Leftrightarrow   \mathcal{P} \succ 0,\\\label{liner}&\delta_2\mathcal{P} - (\mathcal{P} \mathcal{A}_1)^\top  \mathcal{P}^{-1} ( \mathcal{P} \mathcal{A}_2 ) \succeq 0 \Leftrightarrow \begin{bmatrix}
				\mathcal{P} &  \mathcal{P} \mathcal{A}_2\\
				( \mathcal{P}\mathcal{A}_1)^\top &  \delta_2  \mathcal{P}
			\end{bmatrix} \succeq 0.
		\end{align}
		Similarly, applying the Schur complement to the matrix inequalities in~\eqref{matrix-inequality1c},~\eqref{matrix-inequality1d}, and \eqref{matrix-inequality1E} reveals that
		\begin{align}\label{liner1}
			&\mathcal{P} \succ 0, \mathcal{A}_2^\top  \mathcal{P} \mathcal{A}_2  - \delta_3 \mathcal{P} \preceq 0 \Leftrightarrow \begin{bmatrix}
				\mathcal{P} &  \mathcal{P} \mathcal{A}_2\\
				( \mathcal{P}\mathcal{A}_2)^\top & \delta_3 \mathcal{P}
			\end{bmatrix} \succeq 0,
			\end{align}
			\begin{align}\label{liner2}
			&\mathcal{P} \succ 0, \mathcal{A}_2^\top  \mathcal{P} \mathcal{A}_1  - \delta_4 \mathcal{P} \preceq 0 \Leftrightarrow \begin{bmatrix}
				\mathcal{P} &  \mathcal{P} \mathcal{A}_1\\
				( \mathcal{P}\mathcal{A}_2)^\top & \delta_4 \mathcal{P}
			\end{bmatrix} \succeq 0,\\\notag
			&\mathcal{P} \succ 0, (\frac{1-q_\phi}{q_\phi})\mathcal{A}_3^\top  \mathcal{P} \mathcal{A}_3  - \delta_5 \mathcal{P} \preceq 0 \\\label{liner3}&~~~~~~~~~~~~~~~~~~~~~~~~~~~~~~~~\Leftrightarrow \begin{bmatrix}
				\mathcal{P} &  \mathcal{P} \mathcal{A}_3\\
				( \mathcal{P}\mathcal{A}_3)^\top & \frac{\delta_5q_\phi}{1-q_\phi} \mathcal{P}
			\end{bmatrix} \succeq 0.
		\end{align}
	\end{subequations}
	It is worth noting that $ \frac{\delta_5q_\phi}{1-q_\phi}\neq 0$, for $q_\phi \in (0,1)$. This approach eliminates the key bilinear terms in \eqref{matrix-inequality1b} and \eqref{matrix-inequality1d}, which arise from the bilinearity between $\mathcal{P}$ and $\mathcal{F}$. Similarly, the bilinear terms in \eqref{matrix-inequality1c} and \eqref{matrix-inequality1E}, resulting from the product of $\mathcal{F}^\top$ and $\mathcal{F}$ in $\mathcal{A}_2$ and $\mathcal{A}_3$, can also be eliminated. To systematically enforce the satisfaction of the matrix inequality in~\eqref{matrix-inequlity} while computing both $\mathcal{P}$ and $\mathcal{F}$, we propose Algorithm~\ref{Alg1}, which summarizes the required steps. Having addressed the bilinearity issue in Assumption~\ref{assum}, we now present the following theorem as the central contribution of this work, enabling the design of a CBC and its associated safety controller for the dt-SLS-$\Sigma$.
	\begin{algorithm}[t]
		\caption{Computation of $\mathcal{P}$ and $\mathcal{F}$}
		\begin{algorithmic}[1]\label{Alg1}
			\REQUIRE dt-SLS-$\Sigma=(X,U,A,B,w, \tau, \theta, \phi)$, $p_\theta$, $q_\phi$, $\tau$
			\STATE Initialize $\delta = \{\delta_1, \delta_2, \delta_3, \delta_4, \delta_5\}$ such that $\sum_{i=1}^{5}\delta_i = 1$\label{Step2}
			\IF{Matrix inequality in \eqref{matrix-inequality1a} using $\delta$ (via the solver \textsf{Mosek}~\cite{mosek}) is feasible}\label{Step3}
			\RETURN Barrier matrix $\mathcal{P}$\label{Step4}
			\IF{Matrix inequalities in \eqref{liner}--\eqref{liner3} are solved using $\delta$ and $\mathcal{P}$ in step~\ref{Step3}}\label{Step5}
			\RETURN Controller gain $\mathcal{F}$
			\ENDIF
			\ENDIF
			\STATE If either Step~\ref{Step3} or Step~\ref{Step5} is not successful, repeat the algorithm from Step~\ref{Step2} with different $\delta$
			\ENSURE Barrier matrix $\mathcal{P}$ and controller matrix $\mathcal{F}$
		\end{algorithmic}
\end{algorithm}
\begin{theorem}[\textbf{CBC and safety controller}]\label{main-theorem}
	Consider the dt-SLS-$\Sigma=(X,U,A,B,w, \tau, \theta, \phi)$ described~in~Definition~\ref{dt-SLS}. Let Assumption~\ref{assum} hold, and suppose there exist $\beta,\eta \in \mathbb{R}^{+}$ with $ \beta > \eta $ such that
	\begin{subequations}\label{Opt-cons}
		\begin{align}\label{Th:con1}
			& \eta=\lambda_{\max }(\mathcal{P}) \max _{\mathcal{Z} \in X_0 \times X_0 \times U \times U}\Vert \mathcal{Z}   \Vert^2, \\\label{Th:con2}
			& \beta=\lambda_{\min }(\mathcal{P}) \min _{\mathcal{Z} \in X_1 \times X_1 \times U \times U }\Vert \mathcal{Z}   \Vert^2.
		\end{align}
	\end{subequations}
	\noindent Then $\mathbb{B}(\mathcal{Z}) = \mathcal{Z}^{\top} \mathcal{P}\mathcal{Z}$ is a CBC and $u_k$ in~\eqref{dt-SLSCC} is a safety controller with 
	\begin{align}\label{c-martingle}
		c &= \text{Tr} \left(\mathcal{D}^\top \mathcal{P} \mathcal{ D}\sigma_{\mathbf{w}} \right)\!,
	\end{align}
	where $\sigma_{\mathbf{w}}$ is the covariance matrix of the noise vector $\mathbf{w}$. 
\end{theorem}
\begin{proof}
		We first show that the matrix inequality in~\eqref{matrix-inequlity} implies satisfaction of condition~\eqref{con3}. By applying function $\mathbb{B}$ on the augmented dt-ASLS-$\Lambda$ in~\eqref{aug-system}, one has
		\begin{align*}
			\mathbb{B}(\mathcal{Z}_{k+1}) \!\!= \mathcal{Z}_k^{\top} \mathcal{A}^{\top} {\mathcal{P}} 	\mathcal{A} \mathcal{Z}_k \!+\! 2\mathbf{w}_k^{\top} 	\mathcal{D}^{\top} {\mathcal{P}} \mathcal{A} \mathcal{Z}_k\!+\!\mathbf{w}_k^{\top} \mathcal{D}^{\top} {\mathcal{P}} \mathcal{D}\mathbf{w}_k.
		\end{align*} 
		Taking expectation from $\mathbb{B}$ and since $\mathbf{w}\!\sim\! \mathcal{N}(\mathbf{0}_{\psi},\sigma_{\mathbf{w}})$, with $\sigma_{\mathbf{w}}=\mathsf{diag}(\sigma_{w_{k}},\cdots,\sigma_{w_{k-\tau}})$, we have
		\begin{align}\label{proof1}
			\EE\big[\mathbb{B}(\mathcal{Z}_{k+1}) \mid \mathcal{Z}_k\big] \!=\! \EE\left[\mathcal{Z}_k^{\top} \mathcal{A}^{\top} {\mathcal{P}} \mathcal{A} \mathcal{Z}_k\right] 
			\!+\! \EE\left[\mathbf{w}_k^{\top} \mathcal{D}^{\top} {\mathcal{P}} \mathcal{D}\mathbf{w}_k\right]\!\!.
		\end{align}
		Substituting $\mathcal{A}$ from~\eqref{decopmposition} into~\eqref{proof1} and given that $\EE[\zeta_\phi]=0$, we obtain
		\begin{align}\notag
			\EE\big[\mathbb{B}(&\mathcal{Z}_{k+1}) \mid \mathcal{Z}_k\big] \\\notag
			&= \EE\left[\mathcal{Z}_k^{\top} \mathcal{A}_1^{\top} {\mathcal{P}} \mathcal{A}_2 
			\mathcal{Z}_k\right]+
			\EE\left[\mathcal{Z}_k^{\top} \mathcal{A}_2^{\top} {\mathcal{P}} \mathcal{A}_1 \mathcal{Z}_k\right] \\\notag&~~~+\EE\left[\mathcal{Z}_k^{\top} \mathcal{A}_2^{\top} {\mathcal{P}} \mathcal{A}_2 
			\mathcal{Z}_k\right]+\EE\left[\mathcal{Z}_k^{\top} \zeta_\phi^2 \mathcal{A}_3^{\top} {\mathcal{P}} \mathcal{A}_3 
			\mathcal{Z}_k\right]\\\notag
			&~~~+\EE\left[\mathcal{Z}_k^{\top} \mathcal{A}_1^{\top} {\mathcal{P}} \mathcal{A}_1 \mathcal{Z}_k\right]+\text{Tr}\big(\mathcal{D}^{\top} {\mathcal{P}} \mathcal{D} \EE\left[\mathbf{w}_k\mathbf{w}_k^{\top}\right]\big)\\\notag
			&=\mathcal{Z}_k^{\top}\Big(\mathcal{A}_{1}^\top \mathcal{P}\mathcal{A}_{1} +\mathcal{A}_{1}^\top \mathcal{P}\mathcal{A}_{2}\\\notag&~~~+  \mathcal{A}_{2}^\top \mathcal{P}\mathcal{A}_{2} + \mathcal{A}_{2}^\top \mathcal{P}\mathcal{A}_{1}+\EE[\zeta_\phi^2]\mathcal{A}_{3}^\top  \mathcal{P}\mathcal{A}_{3}  \Big)\mathcal{Z}_k\\\notag &~~~ +\text{Tr}\big(\mathcal{D}^{\top} {\mathcal{P}} \mathcal{D} \underbrace{\EE\left[\mathbf{w}_k\mathbf{w}_k^{\top}\right]}_{\sigma_{\mathbf{w}}}\big).
		\end{align}
		By employing the matrix inequalities in~\eqref{matrix-inequality1} and since $\textsf{Var}[\zeta_\phi]=\EE[\zeta_\phi^2]=\frac{1-q_\phi}{q_\phi}$, one can readily show that
		\begin{equation*}
			\EE\Big[ \mathbb{B}(\mathcal{Z}_{k+1}) \big| \mathcal{Z}_k \Big] \!\leq\! \mathbb{B}(\mathcal{Z}_k)+c.
		\end{equation*}
		with $c $ given in~\eqref{c-martingle}.
		To finalize the proof, we now show conditions~\eqref{con} and~\eqref{con1}. It is well known that  
		\[
		\mathbb{B}(\mathcal{Z}) = \mathcal{Z}^{\top} \mathcal{P} \mathcal{Z} \leq \lambda_{\max}(\mathcal{P}) \Vert \mathcal{Z}   \Vert^2.
		\]
		Additionally, \(\forall \mathcal{Z} \in X_0  \!\times\! X_0 \!\times\! U \!\times\! U \), we have  
		\[
		\lambda_{\max}(\mathcal{P}) \Vert \mathcal{Z}   \Vert^2 \leq \lambda_{\max}(\mathcal{P}) \max_{\mathcal{Z} \in  X_0 \times X_0\times U \times U } \Vert \mathcal{Z}   \Vert^2,
		\]
		which confirms the validity of condition~\eqref{con} with  
		\[
		\eta = \lambda_{\max}(\mathcal{P}) \max_{\mathcal{Z} \in  X_0 \times X_0 \times U \times U } \Vert \mathcal{Z}   \Vert^2,
		\]
		as defined in~\eqref{Th:con1}. Similarly, applying the same reasoning, we conclude that condition~\eqref{con1} holds with  
		\[
		\beta=\lambda_{\min }(\mathcal{P}) \min _{\mathcal{Z} \in X_1 \times X_1 \times U \times U} \Vert \mathcal{Z}   \Vert^2,
		\]
		as defined in~\eqref{Th:con2}, which completes the proof. 
\end{proof}

Given that the system evolves probabilistically, in the following theorem, and under the fundamental results in~\cite{ref26}, we employ Definition~\ref{barrier} and obtain an upper bound for the probability that trajectories of dt-SLS-$\Sigma$ reach an unsafe region within a finite time horizon.
\begin{theorem}[\textbf{Safety guarantee under delay and data loss}]\label{Kushner}
	Consider a dt-SLS-$\Sigma=(X,U,A,B,w, \tau, \theta, \phi)$ as in Definition~\ref{dt-SLS} subject to delay and data loss. Let there exist a CBC $\mathbb{B}$ for the augmented system dt-ASLS-$\Lambda$~in~\eqref{aug-system} as in Definition~\ref{barrier}. Then the probability that a trajectory $\mathsf{x}_k$ of dt-SLS-$\Sigma$, starting from any initial state $x_0 \in X_0$ under an input signal $u_k$ and process noise $w_k$ reaches $X_1$ within a time horizon $k \in [0,\mathcal T]$ is
	\begin{equation}\label{zet}
		\PP \Big\{\mathsf{x}_k \in X_1~\text{for some}~ k \in [0,\mathcal T] \,\big|\, x_0\Big\} \leq \xi_{\tau\theta\phi},
	\end{equation} 
	where $\xi_{\tau\theta\phi} = \frac{\eta + c \mathcal T}{\beta}.$
\end{theorem}
\begin{proof}
		According to condition (\ref{con1}), we know $$X_{1}\times X_1 \times U \times U \subseteq \big\{\mathcal{Z} \in X\times X \times U \times U \enspace |\enspace  \mathbb{B}(\mathcal{Z}) \geq \beta \big\}.$$ We, therefore, have
		\begin{align*}  
			\PP &\Big\{\mathsf{x}_k \in X_1\!\times\! X_1\times U \times U ~\text{for some}~ k \in [0,\mathcal T] \,\big|\, x_0 \Big\}\\&~~~  \leq \PP\Big\{\sup_{0\leq k\leq \mathcal T} \enspace \mathbb{B}(\mathcal{Z}) \geq \beta \,\big|\, x_0 \Big\}.
		\end{align*}
		\noindent Then the proposed bound in (\ref{zet}) can be obtained by applying \cite[Theorem 3]{ref26} to (\ref{zet}) and employing conditions~\eqref{con} and~\eqref{con3} in Definition~\ref{barrier}.$\hfill$
\end{proof}

Under Theorem~\ref{Kushner}, the satisfaction of the safety specification $\Upsilon$ within the same time horizon $\mathcal T$ can be guaranteed as
\begin{align}\label{Prob-bound}
	\PP \Big\{\Sigma \models\Upsilon\Big\}\ge 1-\underbrace{\frac{\eta + c \mathcal T}{\beta}}_{\xi_{\tau\theta\phi}}.
\end{align}

The procedure for constructing the CBC and its controller under communication constraints is outlined in Algorithm~\ref{Alg2}.\vspace{0.2cm}

\begin{algorithm}[t]
	\caption{CBC and safety controller construction}
	\label{Alg2}
\begin{algorithmic}[1]
		\REQUIRE dt-SLS-$\Sigma=(X,U,A,B,w, \tau, \theta, \phi)$, $p_\theta$, $q_\phi$, $\tau$, $\Upsilon = (X_0, X_1, \mathcal T)$
		\STATE Compute barrier matrix $\mathcal{P}$ and controller matrix $\mathcal{F}$ using Algorithm~\ref{Alg1}
		\STATE Using the constructed $\mathcal{P}$, compute $c$ as defined in \eqref{c-martingle}
		\STATE Design level sets $\eta, \beta$ in conditions \eqref{Th:con1} and \eqref{Th:con2} using the constructed $\mathbb{B}(\mathcal{Z})$
		\STATE Compute the lower-bound probability $1-\frac{\eta + c \mathcal{T}}{\beta}$
		\ENSURE CBC $\mathbb{B}(\mathcal{Z}) = \mathcal{Z}^{\top}\mathcal{P}\mathcal{Z}$, safety controller $u_k$, and probabilistic safety guarantee~in~\eqref{Prob-bound} 
	\end{algorithmic}
\end{algorithm}
\vspace{-3mm}

\section{Simulation Results}\label{sec:sim}
We illustrate the effectiveness of our proposed framework by applying it to an RLC circuit. We consider communication channels with the probabilities of \emph{successful} data transmission for the uplink and downlink channels given by $p_{\theta} = 0.93$ and $q_{\phi} = 0.90$, while the delay is $\tau = 3$. The RLC circuit dynamics are described as~\cite{9682889}
\begin{equation*}
		\begin{cases}
			x_{1(k+1)}=x_{1k}+\Delta\left(-\frac{R}{L} x_{1k}-\frac{1}{L} x_{2k}\right)+u_{1k} + w_{1k}, \\
			x_{2(k+1)}=x_{2k}+\Delta\left(\frac{1}{C} x_{1k}\right)+u_{2k}+w_{2k},
		\end{cases}
\end{equation*}
where $x_{1k}$ and $x_{2k}$ represent the current and voltage of the circuit, respectively, $\Delta = 0.05$ is the sampling time, $R = 2$ is the series resistance, $L = 9$ is the series inductance, $C = 0.5$ is the circuit capacitance, and $w_{1k},w_{2k} \sim \mathcal{N}(0,0.1)$.
\footnotetext[1]{Due to high dimension of $\mathcal{P} \in \mathbb{R}^{30 \times 30}$, we omit its representation for brevity.}

\begin{figure}[htbp]
	\centering
	\begin{subfigure}[b]{0.48\linewidth}
		\includegraphics[width=\linewidth]{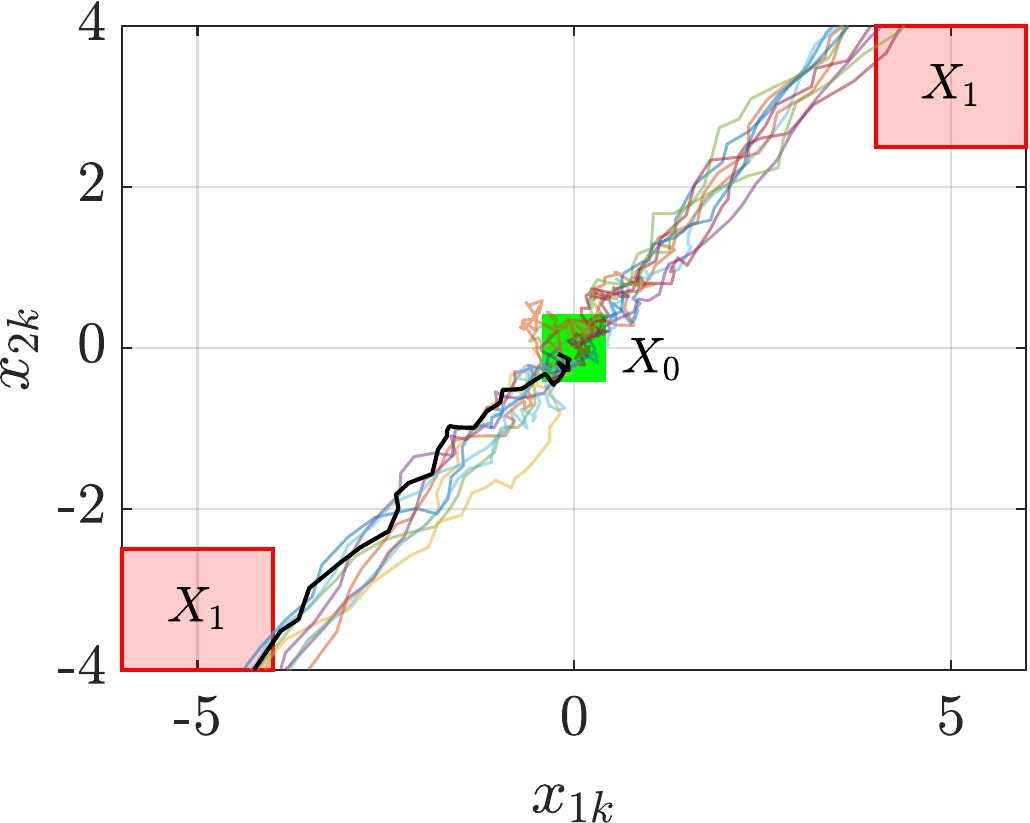}
		\subcaption{Open-loop trajectories}
		\label{fig:open-loop}
	\end{subfigure}
	\hfill
	\begin{subfigure}[b]{0.48\linewidth}
		\includegraphics[width=\linewidth]{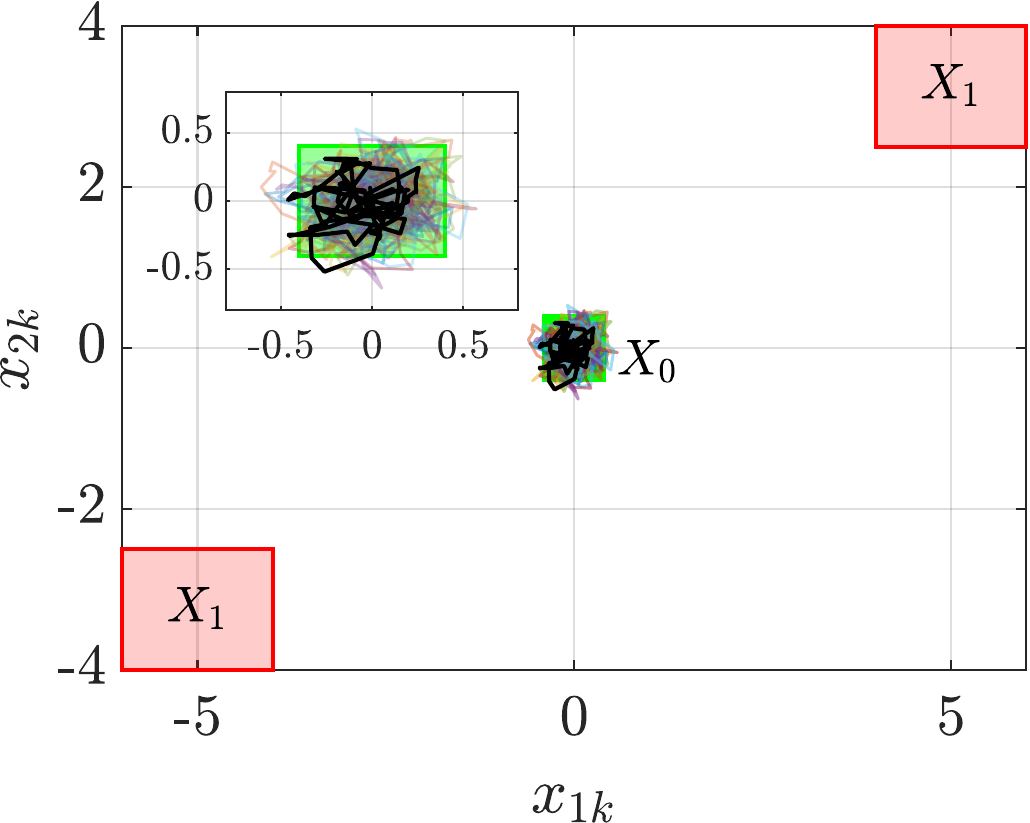}
		\subcaption{Closed loop with $\tau=1$}
		\label{fig:tau-1}
	\end{subfigure}
	\vskip\baselineskip
	\begin{subfigure}[b]{0.48\linewidth}
		\includegraphics[width=\linewidth]{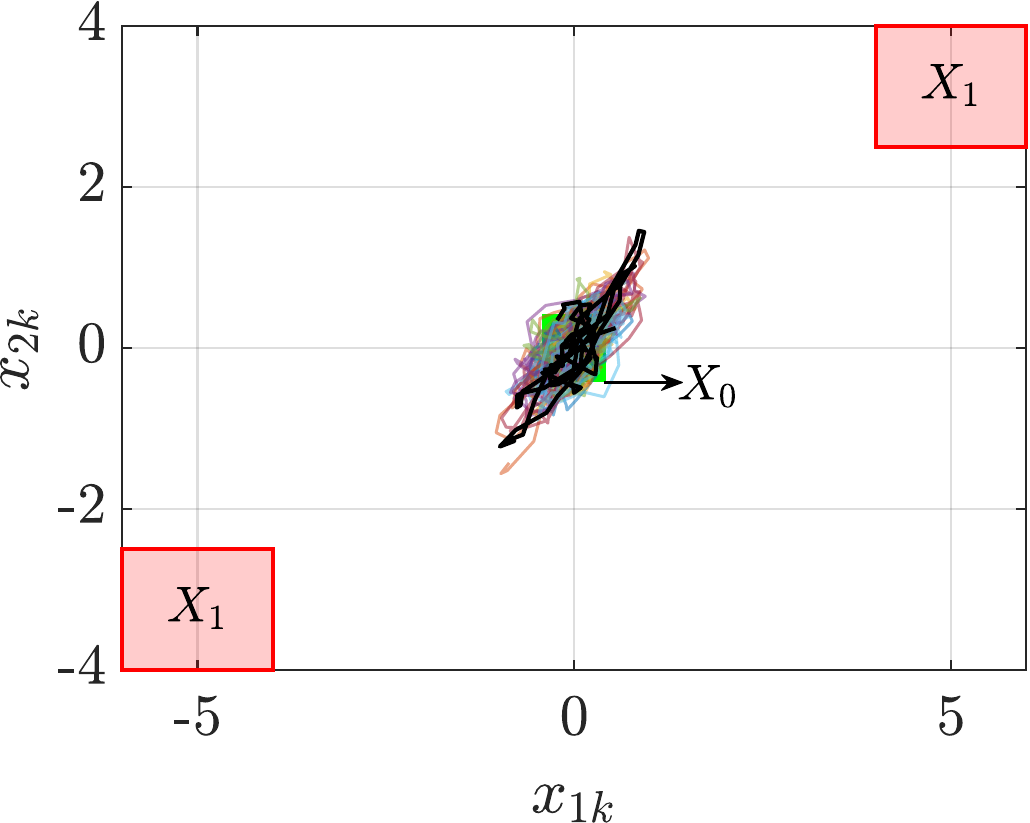}
		\subcaption{Closed loop with $\tau=3$}
		\label{fig:tau-3}
	\end{subfigure}
	\hfill
	\begin{subfigure}[b]{0.48\linewidth}
		\includegraphics[width=\linewidth]{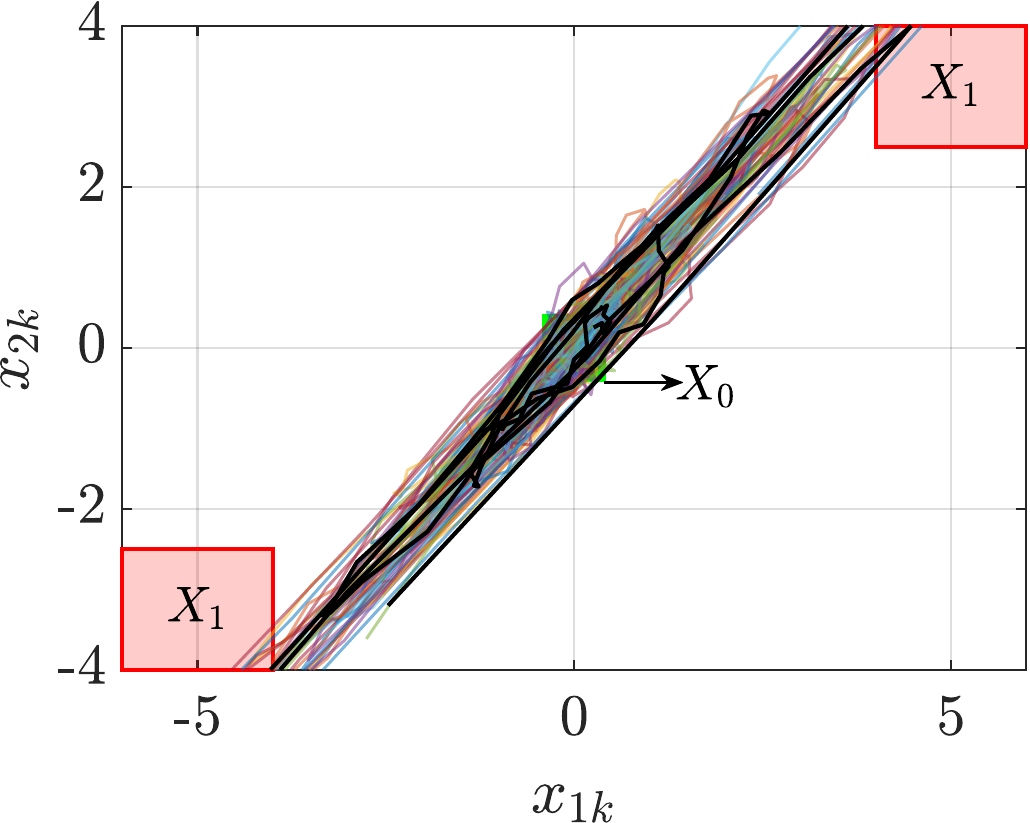}
		\subcaption{Closed loop with $\tau=4$}
		\label{fig:tau-4}
	\end{subfigure}
	
	\caption{
		Trajectories of the RLC circuit over 20 simulation runs with distinct noise realizations starting from the initial set $X_0=\left[-0.4,\,0.4\right]^2$. Initial and unsafe regions are depicted by green \protect\greensquare\ and red \protect\reddsquare\ boxes, respectively. Plot (a) shows the \emph{open-loop} trajectories, while the remaining subplots show closed-loop trajectories with the proposed controller $\mathcal{F}$ with delays of $\tau=1$~(b) and $\tau=3$~(c). Plot (d) illustrates trajectories with delay $\tau = 4$ over $\mathcal{T}=100$, with the controller gain~\eqref{designed-controller} designed for $\tau \leq 3$, showing violation of the safety property $\Upsilon$, as expected.}
	\label{fig:states}
\end{figure}
The regions of interest are given as $X = [-6,6]\times [-4,4]$, $X_0 = [-0.4, 0.4]^2$, and $X_1 = [-6, -4] \times [-4, -2.5] \cup [4, 6] \times [2.5, 4]$. The aim is to synthesize a CBC and a corresponding safety controller for the RLC circuit, ensuring that the RLC circuit's states remain within the safe set $X \backslash X_1$. To this end, following Algorithm~\ref{Alg1}, we initially use the solver \textsf{Mosek} to compute a barrier matrix $\mathcal{P}$\footnotemark[1] and its associated controller gain $\mathcal{F}$ as
\begin{equation}\label{designed-controller}
	\mathcal{F}=\begin{bmatrix}
		-0.2634 & -0.09317\\ -0.09047 & -0.2761 \end{bmatrix}\!\!.
\end{equation}

With the obtained $\mathcal{P}$ and under Algorithm~\ref{Alg2}, we proceed to determine the corresponding initial and unsafe level sets $\eta$ and $\beta$, by satisfying conditions~\eqref{Th:con1} and ~\eqref{Th:con2}, as $\eta =  0.0228$ and $\beta =0.3131 $. It is worth noting that condition~\eqref{Th:con2} should be satisfied separately for both unsafe regions. Furthermore, the constant $c$ is calculated as $7.6307 \times 10^{-5}$. Applying Theorem~\ref{Kushner}, we can guarantee that all trajectories of the RLC circuit with uplink delay of $\tau=3$ initialized from $X_0 = [-0.4, 0.4]^2$, remain within the safe set $X \backslash X_1$ over the time horizon $\mathcal{T} = 100$ with $\PP \Big\{\Sigma \models\Upsilon\Big\}\ge 0.9$. 

As shown in Fig.~\ref{fig:states}, the control gain $\mathcal{F}$ maintains system safety for uplink delays $\tau \leq 3$. In fact, no violations of the safety specification $\Upsilon$ were observed during simulation runs with various noise realizations, for uplink delays $\tau \leq 3$ over the time horizon $\mathcal{T} = 100$. However, the safety specification is violated when the uplink delay $\tau$ reaches $4$ or higher. Fig.~\ref{fig:data-loss} displays the number of lost state and control data corresponding to each of the $20$ noise realizations over the time horizon $\mathcal{T} = 100$. 

It is worth noting that the controller in~\eqref{designed-controller} is designed for a delay of $\tau=3$ and therefore cannot satisfy the system’s safety specification when subjected to delays $\tau>3$, as illustrated in Fig.~\ref{fig:states}, subplot (d). Nevertheless, for delays $\tau>3$ and different values of $q_{\phi}$ and $p_{\theta}$, one may attempt to re-solve the condition in~\eqref{matrix-inequlity} together with those in~\eqref{Opt-cons} to synthesize safety controllers and their corresponding barrier certificates. We also note that the role of the downlink transmission success probability, $q_\phi$, in condition~\eqref{matrix-inequlity} shows that increasing $q_\phi$ improves the likelihood of satisfying condition~\eqref{matrix-inequlity}, as expected. A similar effect is anticipated for the uplink transmission success probability, $p_\theta$.
\begin{figure}[tp!]
	\centering
	\includegraphics[height=0.38\linewidth]{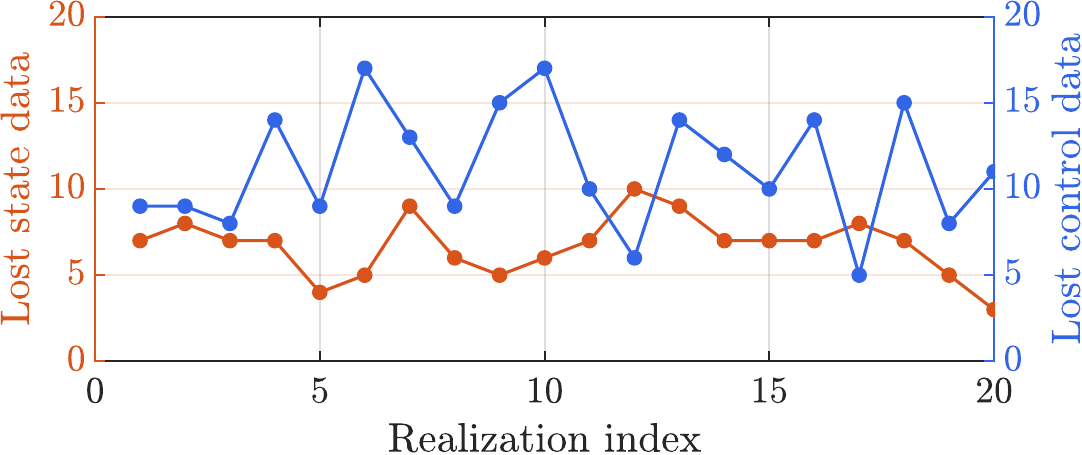}
	\caption{The number of lost states ($\theta_k = 0$) and lost control ($\phi_k = 0$) data for RLC circuit with $\tau=3$ over a time horizon of $\mathcal{T} = 100$ in the uplink (\textbf{orange}) and downlink (\textbf{blue}) channels for $20$ different simulation runs.}
	\label{fig:data-loss}\vspace{-2mm}
\end{figure}

\section{Conclusion}\label{sec:conclusion}
We proposed a framework for analyzing the safety of dt-SLS experiencing communication imperfections due to wireless network constraints. Specifically, we accounted for constant delays in the uplink channel and data loss modeled as independent Bernoulli processes in both uplink and downlink channels. By introducing an augmented dt-ASLS, we effectively captured the dynamics of the original system under delay and data loss constraints. We employed CBCs to design controllers capable of ensuring safety specifications in a probabilistic sense, translating safety requirements into matrix inequalities. We quantified probabilistic guarantees for maintaining system safety despite communication imperfections. We demonstrated our framework through simulations involving an RLC circuit. Extending our framework to accommodate networked control systems with nonlinear dynamics is under investigation as future work.

\bibliographystyle{IEEEtran} 
\bibliography{biblio}

\appendix

	\section{Appendix}
	\noindent The entries of the matrix $\mathcal{A}$ in~\eqref{aug-system} are defined as
	\begin{align}\notag
		\mathcal{A}_{11} &= \begin{bmatrix} \boldsymbol{0}_{n \times (\psi-n)}  & A^{\tau+1}\\
			\mathbf{I}_{(\psi-n)} & \mathbf{0}_{(\psi-n) \times n}
		\end{bmatrix}_{\psi \times \psi}\!\!, \quad
		\mathcal{A}_{12} =  \mathbf{0}_{\psi \times 2n}\\\notag
		\mathcal{A}_{13} &= \begin{bmatrix}  B \!\!&\!\! \mathbf{0}_{n\times m} \!\!&\!\!  \bar{\phi} AB \!\!&\!\! \dots\!\!&\!\!  \bar{\phi} A^{\tau}B\\
			\mathbf{0}_{(\psi-n) \times m} \!\!&\!\! \mathbf{0}_{(\psi-n) \times m} \!\!&\!\! \mathbf{0}_{(\psi-n) \times m} \!\!&\!\! \dots\!\!&\!\! \mathbf{0}_{(\psi-n) \times m}
		\end{bmatrix} \\\notag
		\mathcal{A}_{14} &= \begin{bmatrix}  \mathbf{0}_{n \times m}  \!\!&\!\!  \phi AB \!\!&\!\! \dots\!\!&\!\!  \phi A^{\tau}B\\
			\mathbf{0}_{(\psi-n) \times m} \!\!&\!\! \mathbf{0}_{(\psi-n) \times m} \!\!&\!\! \dots\!\!&\!\! \mathbf{0}_{(\psi-n) \times m}
		\end{bmatrix}_{\psi \times \varpi}\\\notag
		\mathcal{A}_{21} &= \begin{bmatrix} \boldsymbol{0}_{n \times (\psi-n)} & p_\theta A^{\tau+1} \\
			\mathbf{0}_{n \times (\psi-n)} & 	\mathbf{0}_{n \times n}
		\end{bmatrix}_{2n \times \psi}\\\notag
		\mathcal{A}_{22} &= \begin{bmatrix} \boldsymbol{0}_{n \times n} & (1-p_{\theta}) A^{2} \\
			\mathbf{I}_{n} & \mathbf{0}_{n \times n}
		\end{bmatrix}_{2n \times 2n}\!\!,
		\quad \mathcal{A}_{23} =  \mathbf{0}_{2n \times (\varpi+m)}\\\notag
		\mathcal{A}_{24} &= \begin{bmatrix}  B &   AB & p_\theta A^{2}B  & \dots& p_\theta A^{\tau}B\\
			\mathbf{0}_{n \times m}& \mathbf{0}_{n \times m}  & \mathbf{0}_{n \times m} & \dots& \mathbf{0}_{n \times m}
		\end{bmatrix}_{2n \times \varpi} \\\notag
		\mathcal{A}_{31} &= \begin{bmatrix}
			\mathbf{0}_{m \times (\psi-n)}  & \phi p_\theta  \mathcal{F}A^{\tau+1} \\
			\mathbf{0}_{\varpi \times (\psi-n)} & \mathbf{0}_{\varpi \times n} 
		\end{bmatrix}_{(\varpi+m) \times \psi}\\\notag
		\mathcal{A}_{32} &= \begin{bmatrix}
			\mathbf{0}_{m \times n} & \phi (1-p_\theta)\mathcal{F}A^2  \\
			\mathbf{0}_{\varpi \times n} & 	\mathbf{0}_{\varpi \times n}
		\end{bmatrix}_{(\varpi+m) \times 2n}\\\notag
		\mathcal{A}_{33} &= \begin{bmatrix}
			\bar{\phi}\begin{bmatrix}\mathbf{I}_{m} &\mathbf{0}_{m \times (\varpi-m)}\end{bmatrix}_{m\times\varpi}  & \mathbf{0}_{m \times m}  \\
			\mathbf{I}_{\varpi \times \varpi} & 	\mathbf{0}_{\varpi \times m}
		\end{bmatrix}_{(\varpi+m) \times (\varpi+m)}\\\notag
		\mathcal{A}_{34} \!&\!= \begin{bmatrix}  \phi \mathcal{F}B \!\!&\!\!  \phi\mathcal{F}AB \!\!&\!\! \phi p_\theta \mathcal{F}A^{2}B  \!\!&\!\! \dots\!\!&\!\! \phi p_\theta \mathcal{F}A^{\tau}B\\
			\mathbf{0}_{\varpi \times m} \!\!&\!\! \mathbf{0}_{\varpi \times m}  \!\!&\!\! \mathbf{0}_{\varpi \times m} \!\!&\!\! \dots\!\!&\!\! \mathbf{0}_{\varpi \times m}
		\end{bmatrix}\\\notag
		\mathcal{A}_{41} \!&\!= \begin{bmatrix} \boldsymbol{0}_{m \times (\psi-n)} \!\!&\!\! p_\theta \mathcal{F}A^{\tau+1} \\
			\mathbf{0}_{(\varpi-m) \times (\psi-n)} & 	\mathbf{0}_{(\varpi-m)  \times n}
		\end{bmatrix}_{\varpi \times \psi}\\\notag
		\mathcal{A}_{42} &= \begin{bmatrix} \boldsymbol{0}_{m \times n} & (1-p_{\theta}) \mathcal{F}A^{2} \\
			\boldsymbol{0}_{(\varpi-m) \times n} & \mathbf{0}_{(\varpi-m)  \times n}
		\end{bmatrix}_{\varpi \times 2n}\!\!,
		\quad 
		\mathcal{A}_{43} \!=\!  \mathbf{0}_{\varpi \times (\varpi+m)}\\
		\mathcal{A}_{44} \!&\!= \begin{bmatrix}  \begin{bmatrix}\mathcal{F}B \!\!&\!\!  \mathcal{F}AB \!\!&\!\! p_\theta \mathcal{F}A^{2}B\!\!&\!\! \dots\!\!&\!\! p_\theta \mathcal{F}A^{(\tau-1)}B\end{bmatrix}  \!\!&\!\! p_\theta \mathcal{F}A^{\tau}B\\\label{A-elements}
			\mathbf{I}_{(\varpi-m) \times (\varpi-m)} \!\!&\!\! \mathbf{0}_{(\varpi-m)\times m}
		\end{bmatrix}
	\end{align}
	The new entries $\mathcal A^a_{ij}$ and $\mathcal A^b_{ij}$ in \eqref{Newjh} are defined as
	\begin{align}\notag
		\mathcal{A}_{13}^a &= \begin{bmatrix}  B \!\!&\!\! \mathbf{0}_{n\times \varpi}\\
			\mathbf{0}_{(\psi-n) \times m} \!\!&\!\! \mathbf{0}_{(\psi-n) \times \varpi}
		\end{bmatrix}_{\psi \times (\varpi+m)} \\\notag
		\mathcal{A}_{13}^b &= \begin{bmatrix}  \mathbf{0}_{n\times 2m} \!\!&\!\!   AB \!\!&\!\! \dots\!\!&\!\!A^{\tau}B\\
			\mathbf{0}_{(\psi-n) \times 2m} \!\!&\!\! \mathbf{0}_{(\psi-n) \times m} \!\!&\!\! \dots\!\!&\!\! \mathbf{0}_{(\psi-n) \times m}
		\end{bmatrix}_{\psi \times (\varpi+m)} \\\notag
		\mathcal{A}_{14}^a &= \begin{bmatrix}  \mathbf{0}_{n \times m}  \!\!&\!\!  AB \!\!&\!\! \dots\!\!&\!\!  A^{\tau}B\\
			\mathbf{0}_{(\psi-n) \times m} \!\!&\!\! \mathbf{0}_{(\psi-n) \times m} \!\!&\!\! \dots\!\!&\!\! \mathbf{0}_{(\psi-n) \times m}
		\end{bmatrix}_{\psi \times \varpi} \\\notag
		\mathcal{A}_{31}^a &= \begin{bmatrix}
			\mathbf{0}_{m \times (\psi-n)}  & p_\theta  \mathcal{F}A^{\tau+1} \\
			\mathbf{0}_{\varpi \times (\psi-n)} & \mathbf{0}_{\varpi \times n} 
		\end{bmatrix}_{(\varpi+m) \times \psi}\\\notag
		\mathcal{A}_{32}^a &= \begin{bmatrix}
			\mathbf{0}_{m \times n} &  (1-p_\theta)\mathcal{F}A^2  \\
			\mathbf{0}_{\varpi \times n} & 	\mathbf{0}_{\varpi \times n}
		\end{bmatrix}_{(\varpi+m) \times 2n}\\\notag
		\mathcal{A}^a_{33} &= \begin{bmatrix}
			\mathbf{0}_{m\times\varpi}  & \mathbf{0}_{m \times m}  \\
			\mathbf{I}_{\varpi \times \varpi} & 	\mathbf{0}_{\varpi \times m}
		\end{bmatrix}_{(\varpi+m) \times (\varpi+m)}\\\notag
		\mathcal{A}^b_{33} &= \begin{bmatrix}
			\begin{bmatrix}\mathbf{I}_{m} &\mathbf{0}_{m \times (\varpi-m)}\end{bmatrix}_{m\times\varpi}  & \mathbf{0}_{m \times m}  \\
			\mathbf{0}_{\varpi \times \varpi} & 	\mathbf{0}_{\varpi \times m}
		\end{bmatrix}_{(\varpi+m) \times (\varpi+m)}\\
		\mathcal{A}_{34}^a \!&\!= \begin{bmatrix}   \mathcal{F}B \!\!&\!\!  \mathcal{F}AB \!\!&\!\! p_\theta \mathcal{F}A^{2}B \!\!&\!\! \dots\!\!&\!\! p_\theta \mathcal{F}A^{\tau}B\\\label{New-entries}
			\mathbf{0}_{\varpi \times m}\!\!&\!\! \mathbf{0}_{\varpi \times m}  \!\!&\!\! \mathbf{0}_{\varpi \times m} \!\!&\!\! \dots\!\!&\!\! \mathbf{0}_{\varpi \times m}
		\end{bmatrix}
\end{align}
\end{document}